\documentclass[conference]{IEEEtran}
\makeatletter
\def\ps@headings{%
\def\@oddhead{\mbox{}\scriptsize\rightmark \hfil \thepage}%
\def\@evenhead{\scriptsize\thepage \hfil \leftmark\mbox{}}%
\def\@oddfoot{}%
\def\@evenfoot{}}
\makeatother
\usepackage{times}
\usepackage{algorithmic}
\usepackage{algorithm}
\usepackage{amssymb}
\usepackage{color}
\usepackage{graphicx}
\usepackage{subfig}
\usepackage{hyperref}
\usepackage{makeidx}  
\usepackage{latexsym}
\usepackage{amsmath}
\usepackage{amsfonts}
\usepackage{epsfig}
\usepackage{boxedminipage}
\usepackage{url}

\begin{document}
\pagestyle{headings}


\newtheorem{theorem}{Theorem}[section]
\newtheorem{lemma}[theorem]{Lemma}
\newtheorem{claim}[theorem]{Claim}
\newtheorem{fact}[theorem]{Fact}
\newtheorem{corollary}[theorem]{Corollary}
\newtheorem{remark}[theorem]{Remark}
\newtheorem{definition}{Definition}
\newtheorem{procedure}{Procedure}

\renewcommand{\algorithmicrequire}{\textbf{Input:}}
\renewcommand{\algorithmicensure}{\textbf{Output:}}
\newcommand{\sq}{\hbox{\rlap{$\sqcap$}$\sqcup$}}
\newcommand{\qed}{\hspace*{\fill}\sq}

%

\newcommand{\no}{\noindent}
\newcommand{\bea}{\begin{eqnarray}}
\newcommand{\eea}{\end{eqnarray}}
\newcommand{\beq}{\begin{equation}}
\newcommand{\eeq}{\end{equation}}
\newcommand{\beqs}{\begin{equation*}}
\newcommand{\eeqs}{\end{equation*}}
\newcommand{\beas}{\begin{eqnarray*}}
\newcommand{\eeas}{\end{eqnarray*}}
\newcommand{\ep}{\end{proof}}
\newcommand{\bp}{\begin{proof}}
\newcommand{\D}{\displaystyle}

\newenvironment{statement}{\begin{trivlist}
\item{\bf Statement:}}{\it \end{trivlist}}

\newenvironment{internalproof}{\begin{trivlist}
\item{\bf Proof:}}
{\hfill$\Box$\end{trivlist}}

\def\fct#1{{\mathop{\rm #1}}}   
\def\abs{\fct{abs}}             
\def\sign{\fct{sign}}           
\def\re{\fct{Re\,}}             
\def\im{\fct{Im\,}}             
\def\argmax{\mathop{\rm argmax}}
\def\argmin{\mathop{\rm argmin}}

\def\cS{\mathbb{S}}
\def\hG{{\mathbf{p}}}
\def\hg{{\mathbf{p}}}
\def\tg{\tilde{G}}
\def\tG{\tilde{G}}
\def\tC{\tilde{C}}
\def\tP{\tilde{\Pi}}
\def\tp{\tilde{\Pi}}
\def\tS{\tilde{S}}
\def\hC{\widehat{C}}
\def\hc{\widehat{C}}
\def\wc{\widehat{C}}
\def\bc{\widehat{C}}
\def\bC{\widehat{C}}
\def\hl{\frac{1}{2} \log_2 }
\def\PBS0{\frac{\widehat{h}_b P_T}{M} }
\def\PG{P^g  }
\def\PB{P^b  }
\def\PGS{P^{g^*}  }
\def\PBS{P^{b^*}  }
\def\RG{R^g  }
\def\RB{R^b  }
\def\a{\alpha }
\def\DH{\Delta \! H }
\def\-{  \!-\! }
\def\+{  \!+\! }
\def\argmax{\operatornamewithlimits{arg\,max}}
\def\mbf{\mathbf}
\def\hsA{\hspace*{0.3in}}
\def\hsB{\hspace{0.5cm}}
\def\f{\frac{1}{2}}
\def\o{\overline}
\def\n{\nonumber}
\def\b{\bar}
\def\hP{\hat{P}}
\def\mb{\mbox}

\newcommand{\remove}[1]{}
\def\M{{\mathcal M}}
\def\N{{\cal N}}
\def\H{{\cal N}}
\def\R{{\cal R}}
\def\Q{{\cal Q}}
\def\T{{\cal T}}
\def\p{{\mathbf{p}}}
\def\q{{\mathbf{q}}}
\def\pm{{\mathbf{p_{\min}}}}
\def\C{\widetilde C}
\def\D{\widehat D}
\def\paths{{\cal P}}


\sloppy

%

\title{Polynomial Bottleneck Congestion Games with Optimal Price of Anarchy}

\author{Rajgopal Kannan\\
Dept. of Computer Science\\
Louisiana State University\\
Baton Rouge, LA 70803 \\
\and Costas Busch\\
Dept. of Computer Science\\
Louisiana State University\\
Baton Rouge, LA 70803 \\
\and Athanasios V. Vasilakos\\
Dept. of Telecomm. Engg.\\
Univ. of Western Macedonia\\
Greece\\
}

\date{}

%

\maketitle

\begin{abstract}
We study {\em bottleneck congestion games} where the
social cost is determined by the worst congestion
of any resource.
These games directly relate to network routing problems
and also job-shop scheduling problems.
In typical bottleneck congestion games,
the utility costs of the players
are determined by the worst congested resources that they use.
However, the resulting Nash equilibria
are inefficient, since the price of anarchy
is proportional on the number of resources which can be high.
Here we show that we can get smaller price of anarchy
with the bottleneck social cost metric.
We introduce the {\em polynomial bottleneck games}
where the utility costs of the players
are polynomial functions of the congestion of the resources that they use.
In particular, the delay function for any resource $r$ is $C_{r}^\M$,
where $C_r$ is the congestion measured as the number of players that use $r$,
and $\M \geq 1$ is an integer constant that defines the degree of the polynomial.
The utility cost of a player is the sum of the individual delays of the resources
that it uses.
The social cost of the game remains the same, namely, it is the
worst bottleneck resource congestion: $\max_{r} C_r$.
We show that polynomial bottleneck games are very efficient and give
price of anarchy $O(|R|^{1/(\M+1)})$,
where $R$ is the set of resources.
This price of anarchy is tight, since we demonstrate a
game with price of anarchy $\Omega(|R|^{1/(\M+1)})$, for any $\M \geq 1$.
We obtain our tight bounds by using two proof techniques: {\em transformation},
which we use to convert arbitrary games to simpler games,
and {\em expansion}, which we use to bound the price of anarchy in a simpler game.

\end{abstract}



\section{Introduction}
\label{section:intro}
We consider non-cooperative congestion games with $n$ players,
where each player has a {\em pure strategy profile}
from which it selfishly selects a strategy
that minimizes the player's utility cost function
(such games are also known as {\em atomic} or {\em unsplittable-flow} games).
We focus on {\em bottleneck congestion games} where the objective for the social outcome
is to minimize $C$, the maximum congestion on any resource.
Typically, the congestion on a resource is a non-decreasing function
on the number of players that use the resource;
here, we consider the congestion to be simply
the number of players that use the resource.

Bottleneck congestion games have been studied
in the literature \cite{BO07,BM06,BKV08}
in the context of routing games,
where each player's utility cost is
the worst resource congestion on its strategy.
For any resource $r$, we denote by $C_r$
the number of users that use $r$ in their strategies.
In typical bottleneck congestion games,
each player $i$ has utility cost function
$C_i = \max_{r \in S_i} C_r$,
where $S_i$ is the strategy of the player.
The social cost is worst congested resource:
$C = \max_i C_i = \max_r C_r$.

In \cite{BO07} the authors observe that
bottleneck games are important in networks for various practical reasons.
In networks, each resource corresponds to a network link,
each player corresponds to a packet,
and a strategy represents a path for the packet.
In wireless networks, the maximum congested link is related to the
lifetime of the network since the nodes adjacent to high congestion links
transmit large number of packets which results to
higher energy utilization.
High congestion links also result to congestion hot-spots which
may slow-down the network throughput.
Hot spots also increase the vulnerability of the network
to malicious attacks which aim to to increase the congestion of links
in the hope to bring down the network.
Thus, minimizing the maximum congested edge results to hot-spot avoidance
and more load-balanced and secure networks.

In networks, bottleneck games are also important from a theoretical point of view
since the maximum resource congestion is immediately related to the
optimal packet scheduling.
In a seminal result,
Leighton {\it et al.} \cite{LMR94}
showed that there exist packet scheduling algorithms that
can deliver the packets along their chosen paths in time very close to $C+D$,
where $D$ is the maximum chosen path length.
When $C \gg D$,
the congestion becomes the dominant factor in the packet scheduling performance.
Thus, smaller $C$ immediately implies faster packet delivery time.

A natural problem that arises in games
concerns the effect of the players'
selfishness on the welfare of the whole system
measured with the {\em social cost} $C$.
We examine the consequence of the selfish behavior
in pure {\em Nash equilibria} which are stable states of the game in which no
player can unilaterally improve her situation.
We quantify the effect of
selfishness with the {\em price of anarchy}
($PoA$)~\cite{KP99,P01}, which expresses how much larger is the
worst social cost in a Nash equilibrium compared to the social cost
in the optimal coordinated solution.
The price of anarchy provides a measure for estimating
how closely do Nash equilibria of bottleneck congestion games approximate the
optimal $C^*$ of the respective coordinated optimization problem.

Ideally, the price of anarchy should be small.
However, the current literature results have only provided weak bounds for bottleneck games.
In \cite{BO07} it is shown that if the resource congestion delay function is bounded
by some polynomial with degree $k$ (with respect to the packets that use the resource)
then $PoA = O(|R|^k)$,
where $R$ is the set of links (resources) in the graph.
In \cite{BM06} the authors consider bottleneck routing games
for the case $k=1$ and they show that $PoA = O(L + \log |V|)$,
where $L$ is the maximum path length (maximum number of resources)
in the players' strategies and $V$ is the set of nodes in the network.
This bound is asymptotically tight (within poly-log factors)
since it is shown in \cite{BM06}
that there are game instances with $PoA = \Omega(L)$.
Since $L = |R|$,
the price of anarchy has to be large, $PoA = \Omega(|R|)$.

\subsection{Contributions}

The lower bound in \cite{BM06} suggests that in order to obtain
better price of anarchy in bottleneck congestion games
(where the social cost is the bottleneck resource $C$),
we need to consider alternative player utility cost functions.
Towards this goal, we introduce {\em polynomial bottleneck games}
where the player cost functions
are polynomial expressions of the congestions along the resources.
In particular, the player utility cost function for player $i$ is:
$C'_i = \sum_{r \in S_i} C_{r}^{\M},$
for some integer constant $\M \geq 1$.
Note that the new utility cost is a sum of polynomial
terms on the congestion of the resources in the chosen strategy
(instead of the max that we described earlier).
The social cost remains the maximum bottleneck congestion $C$,
the same as in typical congestion games.

The new player utility costs
have significant benefits in improving
both the upper and lower bounds on the price of anarchy.
For the bottleneck social cost $C$
we prove that the price of anarchy of polynomial games is:
\begin{equation*}
PoA = O(|R|^{1 / (\M + 1)}),
\label{eqn:basic}
\end{equation*}
for any constant $\M \geq 1$.
We show that this bound is asymptotically tight by providing an instance
of a polynomial bottleneck game with $PoA = \Omega(|R|^{1 /(M+1)})$,
for any constant $\M \geq 1$.
Our price of anarchy bound is a significant improvement over the
price of anarchy from the typical bottleneck games described above.


Polynomial congestion games are interesting variations of bottleneck games
not only because they provide good price of anarchy but also
because they represent interesting and important real-life problems.
In networks, the overall delay that a packet experiences
is directly related with the link congestions along the path
and hence the polynomial utility cost function
reflects the total delivery delay.
In wireless networks, the polynomial player utilities correspond
to the total energy that a packet consumes while it traverses the network,
and the social cost reflects to the worst energy utilization in 
any node in the network.
Similar benefits from polynomial congestion games
appear in the context of job-shop scheduling,
where computational tasks require resources to execute.
In this context, the social bottleneck cost function $C$ represents
the task load-balancing efficiency of the resources,
and the player utility costs
relate to the makespan of the task schedule.
In all the above problems, the polynomial degree $\M$
is chosen appropriately to model precisely the involved costs
of the resource utilization in each computational environment.


In our analysis, we obtain the price of anarchy upper bound
by using two techniques: {\em transformation} and {\em expansion}.
Consider a game $G$ with a Nash equilibrium $S$ and congestion $C$.
We identify two kinds of players in $S$:
type-A players which use
only one resource in their strategies,
and type-B players which use two or more resources.
In our first technique, transformation, we convert $G$ to a simpler game $\tilde G$,
having a Nash equilibrium $\tilde S$ with congestion $\tilde C$,
such that $\tilde C = O(C)$, and all players in $\tilde S$ with congestion above a threshold $\tau$
are of type-A; that is, we transform type-B players to type-A players.
Having type-A players is easier to bound the price of anarchy.
Then, we use a second technique, expansion,
which is used to give an upper bound on the price of anarchy of game $\tilde G$,
which implies an upper bound on the price of anarchy of the original game $G$.



In \cite{SAGT}, we have derived upper bounds for the price of anarchy of
games with exponential utility cost functions using similar techniques
(transformation and expansion).  While exponential cost games have
a unique substructure which makes the analysis of Price of Anarchy
much simpler, we believe these techniques are general enough to adapt
in a non-trivial manner for a large class of utility cost functions.
For the case of exponential cost games, we obtained logarithmic price of
anarchy upper bounds, which was related to the problem structure.  Here we
obtain tight (optimal) price of anarchy bounds for polynomial bottleneck
games using a non-trivial application of the general transformation and
expansion techniques.

\subsection{Related Work}

Congestion games were introduced and
studied in~\cite{monderer1,rosenthal1}.
In \cite{rosenthal1}, Rosenthal proves that
congestion games have always pure Nash equilibria.
Koutsoupias and Papadimitriou \cite{KP99} introduced
the notion of price of anarchy
in the specific {\em parallel link networks} model
in which they provide the bound $PoA = 3/2$.
Roughgarden and Tardos \cite{roughgarden3}
provided the first result for splittable flows in general networks
in which they showed that $PoA\le 4/3$ for a player
cost which reflects to the sum of congestions of the resources of a path.
Pure equilibria with atomic flow have been studied in
\cite{BM06,CK05,libman1,STZ04}
(our work fits into this category),
and with splittable flow in
\cite{roughgarden1,roughgarden2,roughgarden3,roughgarden5}.
Mixed equilibria with atomic flow have been studied in
\cite{czumaj1,GLMMb04,KMS02,KP99,LMMR04,MS01,P01},
and with splittable flow in
\cite{correa1,FKS02}.

Most of the work in the literature uses a cost metric
related to the sum of congestions of all the resources of the player's path
\cite{CK05,roughgarden2,roughgarden3,roughgarden5,STZ04}.
In terms of our notation,
the player cost functions are polynomials of degree $\M = 1$.
However, the social cost in those games is different than
ours since it is an aggregate function of the player flows and
congestion of all the resources.
On the other hand, the social cost in our case corresponds to the bottleneck congestion
which is a metric that reduces to a single resource.
The vast majority of the work on congestion games has been performed
for parallel link networks,
with only a few exceptions on network topologies
\cite{BM06,CK05,correa1,roughgarden1}.
Our work immediately applies to network topologies.

In \cite{BM06}, the authors consider bottleneck routing games in networks
with player cost $C_i$ and social cost $C$.
They prove that the price of stability is 1
(the price of stability measures the ratio
of the best Nash equilibrium social cost
versus the coordinated optimal solution).
They show that the price of anarchy is bounded by $O(L + \log |V|)$,
where $L$ is the maximum allowed path length, and $V$ is the set of nodes.
They also prove that $\kappa \leq PoA \leq c(\kappa^2 + log^2 |V|)$,
where $\kappa$ is the size of the largest resource-simple cycle in the graph and $c$ is a constant.
That work was extended in \cite{BKV08} to the $C+D$ routing problem.
Bottleneck routing games have also been studied in~\cite{BO07},
where the authors consider the maximum congestion metric in general networks
with splittable and atomic flow (but without considering path lengths).
They prove the existence and non-uniqueness
of equilibria in both the splittable and atomic
flow models.
They show that finding the best Nash equilibrium that
minimizes the social cost is a NP-hard problem.
Further, they show that the price of anarchy may be unbounded for specific
resource congestion functions.
In \cite{Harks09}, the authors prove the existence of strong Nash equilibria
(which concern coalitions of players)
for games with the lexicographic improvement property;
such games include Bottleneck congestion games and our polynomial games.

%


\subsection*{Outline of Paper}
In Section \ref{section:definitions}
we give basic definitions.
In Section \ref{section:BtoA} we
convert games with type-B players to games with type-A players.
In Section \ref{section:anarchy}
we give a bound on the price of anarchy.
We finish with providing a lower bound in Section \ref{section:lower}.

\section{Definitions}
\label{section:definitions}

A {\em congestion game} is a strategic game
$G = (\Pi_G, R, \cS, (d_r)_{r \in R}, (pc_{\pi})_{\pi \in \Pi_G})$
where:
\begin{itemize}
\item
$\Pi_G = \{\pi_1,\ldots, \pi_n\}$ is a non-empty and finite set of players.

\item
$R = \{r_1,\ldots,r_z\}$ is a non-empty and finite set of resources.

\item
$\cS = \cS_{\pi_1} \times \cS_{\pi_2} \times \cdots \times \cS_{\pi_n}$,
where $\cS_{\pi_i}$ is a strategy set for player $\pi_i$,
such that $\cS_{\pi_i} \subseteq powerset(R)$;
namely, each strategy $S_{\pi_i} \in \cS_{\pi_i}$ is pure,
and it is a collection of resources.
A {\em game state} (or {\em pure strategy profile})
is any $S \in \cS$.
We consider {\em finite games} which have finite $\cS$ (finite number of states).

\item
In any game state $S$,
each resource $r \in R$ has a {\em delay cost} denoted $d_r(S)$.

\item
In any game state $S$,
each player $\pi \in \Pi_G$ has a player cost $pc_\pi(S) = \sum_{r \in S_\pi} d_r(S)$.
\end{itemize}

Consider a game $G$ with a state $S = (S_{\pi_1}, \ldots, S_{\pi_n})$.
The {\em (congestion) of a resource $r$}
is defined as $C_r(S) = |\{ \pi_i: r \in S_{\pi_i} \}|$,
which is the number of players that use $r$ in state $S$.
The {\em (bottleneck) congestion of a set of resources $Q \subseteq R$}
is defined as $C_Q(S) = \max_{r \in Q} C_r(S)$,
which is the maximum congestion over all resources in $Q$.
The \emph{(bottleneck) congestion of state $S$} is denoted $C(S) = C_R(S)$,
which is the maximum congestion over all resources in $R$.
The {\em length of state $S$} is defined to be $L(S) = \max_i |S_{\pi_i}|$,
namely, the maximum number of resources used in any player.
When the context is clear, we will drop the dependence on $S$.
We examine polynomial congestion games:
\begin{itemize}
\item
{\em Polynomial games:}
The delay cost function for any resource $r$ is $d_r = C_r^{\M}$, 
for some integer constant $\M \geq 1$.
\end{itemize}

For any state $S$,
we use the standard notation
$S = (S_{\pi_i},S_{-{\pi_i}})$
to emphasize the dependence on player $\pi_i$.
Player $\pi_i$
is \emph{locally optimal} (or {\em stable}) in state $S$ if
$pc_{\pi_i}(S) \leq pc_{\pi_i}((S'_{\pi_i},S_{-{\pi_i}},))$ for
all strategies $S'_{\pi_i} \in \cS_{\pi_i}$.
A greedy move by a player $\pi_i$ is any change of its strategy
from $S'_{\pi_i}$ to $S_{\pi_i}$
which improves the player's cost,
that is, $pc_{\pi_i}((S_{\pi_i},S_{-{\pi_i}})) < pc_{\pi_i}((S'_{\pi_i},S_{-{\pi_i}}))$.
{\em Best response dynamics} are sequences of greedy moves by players.
A state $S$ is in a {\em Nash Equilibrium}
if every player is locally optimal.
Nash Equilibria quantify the notion of
a stable selfish outcome.
In the games that we study there could exist multiple Nash Equilibria.

For any game $G$ and state $S$,
we will consider a \emph{social cost} (or {\em global cost})
which is simply the bottleneck congestion $C(S)$.
A state $S^*$ is called {\em optimal} if it has minimum
attainable social cost: for any other state $S$, $C(S^*) \le C(S)$.
We will denote $C^* = C(S^*)$.
We quantify the quality of the states which are Nash Equilibria with the
\emph{price of anarchy} ($PoA$)
(sometimes referred to as the coordination ratio).
Let $\cal P$ denote the set of distinct
Nash Equilibria.
Then the price of anarchy of game $G$ is:
\begin{equation*}
PoA(G)
=\sup\limits_{ S \in ~{\cal P}} \frac{C(S)}{C^*},
\end{equation*}

%

We continue with some more special definitions that we use in the proofs.
Consider a game $G$ with a socially optimal state $S^* = (S^*_{\pi_1},
\ldots, S^*_{\pi_n})$, and let $S = (S_{\pi_1}, \ldots, S_{\pi_n})$
denote the equilibrium state.  We consider two special kinds of players
with respect to states $S$ and $S^*$:

\begin{itemize}
\item
{\em Type-A players:}
any player $\pi_i$ with $|S_{\pi_i}| = 1$.

\item
{\em Type-B players:}
any player $\pi_i$ with $|S_{\pi_i}| \geq 2$.
\end{itemize}

For any resource $r \in R$, we will let $\Pi_r$ and $\Pi^*_r$ denote
the set of players with $r$ in their equilibrium and socially optimal
strategies respectively, i.e $\Pi_r = \{ \pi_i \in \Pi_G | r \in S_{\pi_i} \}$
and $\Pi^*_r = \{ \pi_i \in \Pi_G | r \in S^*_{\pi_i} \}$.

Let $G = (\Pi_G, R, \cS, d, (pc_{\pi})_{\pi \in \Pi_G})$ and $\tG =
(\Pi_{\tG}, \tilde{R}, \tilde {\cS}, \tilde{d}, (\tilde{pc}_{\pi})_{\pi
\in \Pi_{\tG}})$ be two games.  We say that {\em $G$ dominates $\tG$}
if the following conditions hold between them for the highest cost
Nash equilibrium and optimal states and : $|\tilde{R}| \leq |R|$,
$d = \tilde{d}$, $\tilde{C} = C$, $C^* \leq \tilde{C}^* \leq  \beta C^* $,
where $\beta >1$ is a constant and $C,C^*$ and $\tilde{C}, \tilde{C}^*$ represent the bottleneck
congestions in the  highest cost Nash equilibrium and optimal states of
$G$ and $\tG$, respectively.

\begin{corollary}
$PoA(G) \leq \beta \cdot PoA(\tG)$ for an arbitrary game $G$ and dominated
game $\tG$.
\label{dominatingcorollary}
\end{corollary}

In the next section, we will describe how an arbitrary game $G$
in Nash equilibrium state $S$ can be transformed into a dominated
game $\tG$ containing type $A$ players of arbitrary cost and type
$B$ players restricted to costs below a given threshold.



\section{Type-$B$ to Type-$A$ Game Transformation}
\label{section:BtoA}

We first state our main results in this section. 

\begin{theorem}
{\it Every game $G$ with highest-cost Nash equilibrium state $S$ can be
transformed into a game $\tilde{G}$ with Nash equilibrium state $\tS$
in which all resources $r$ with congestion $C_r > \psi = \max(2\M,3C^*)$
are occupied exclusively by type-$A$ players. 
}
\label{BtoAlemma}
\end{theorem}
\bigskip

\begin{theorem}
{\it 
$\tG$ is dominated by $G$, i.e the bottleneck congestion in optimal
states $S^*$ and $\tS^*$ of $G$ and $\tilde{G}$ satisfies $C^* \leq \tC^*
\leq 7 C^*$.
}
\label{dominatedtheorem}
\end{theorem} 
\bigskip
\medskip

We prove Theorem~\ref{BtoAlemma} by constructing $\tG$ via the
transformation algorithm below and defer the proof of the domination of
$\tG$ for later.  We first describe some needed preliminaries.

{\it Preliminaries}:
We initialize $\tG$, the input to our transformation algorithm as
a restricted version of game $G$ with exactly two strategies per
player: $\tS_{\pi} = S_{\pi}$ and $\tS_{\pi}^* = S^*_{\pi}$. We 
iteratively transform $\tG$ by converting type-$B$ players of cost
at least $T=\psi^{\M} + 1$ into type-$A$ players, one at a time in
decreasing order of player costs until all type-$B$ players remaining
either fall below the threshold cost function $T$ or no type-$B$ players
exist.  We add and delete players/resources from $\tG$ iteratively and
have a working set of players.  However $\tG$ will always remain in
equilibrium state $\tS$ at every step of the transformation process.
When we add a new player $\pi_k$ to $\tP$ we will assign two strategy
sets to $\pi_k$: an `equilibrium' strategy $\tS_{\pi_k}$ and an optimal
strategy $\tS^*_{\pi_k}$. Thus $\tilde{S} = \tilde{S} \bigcup \tS_{\pi_k}$
and $\tilde{S}^* = \tilde{S}^* \bigcup \tS^*_{\pi_k}$.

We then convert $\tG$ into a `clean' version in which every type-$B$
player $\pi \in \tp$ has {\it distinct} resources in its equilibrium and
optimal strategies i.e $\tS_{\pi} \bigcap \tS^*_{\pi} = \emptyset$. If
not already true, this can be achieved by creating $|\tS_{\pi} \bigcap
\tS^*_{\pi}|$ new type-$A$ players with identical and one type-$B$
player with disjoint equilibrium and optimal strategies for each original
player $\pi$.  The new type-$B$ player has $\tS_{\pi} - \tS^*_{\pi}$
and $\tS^*_{\pi} - \tS_{\pi}$ as its equilibrium and optimal strategy
respectively while the new type-$A$ players each use one resource from
$|\tS_{\pi} \bigcap \tS^*_{\pi}|$ as their identical equilibrium and
optimal strategies. Note that the new players are also in equilibrium
in $\tS$. We also assume throughout that type-$A$ players in $\tG$ have
no redundant resources in their optimal strategies, i.e if a resource
can be removed from $\tS^*_{\pi}$ without affecting $pc_{\pi}(\tS_{\pi})
\leq pc_{\pi}(\tS^*_{\pi})$, then it is removed.


Let $\pi_i$ be an arbitrary type-$B$-player using $k$ resources $r_1, r_2,
\ldots, r_k$ in its equilibrium strategy $\tS_{\pi_i}$ that are distinct
from the $m$ resources $r_1^*, \ldots,  r_m^*$ in its optimal strategy
$\tS^*_{\pi_i}$. Let $C_{r_j}$, $C_{r^*_j}$ denote the congestion on
these resources in equilibrium state $\tS$.
Define procedure $\mathbf{PMS\-Partition} (\pi_i)$ as follows:

\begin{procedure}
Partition $\tS_{\pi_i}$ and $\tS^*_{\pi_i}$ into $t$ pairs $(L_1, L_1^*),
(L_2, L_2^*), \ldots, (L_t, L_t^*)$ where
\begin{enumerate}
\item
The $L_j$'s form a disjoint resource partition of $\tS_{\pi_i}$.  


\item
$L^*_j \subseteq \tS^*_{\pi_i}$ and $|L^*_j \bigcap L^*_{k}| \leq 1$, for $1 \leq j,k \leq t$. 

\item

\beq
\sum_{r \in L^*_j}  (C_{r}+1 )^{\M}  \geq \sum_{r \in L_j} C_r^{\M} , \ \ 1 \leq j \leq t
\label{sumres1}
\eeq
\end{enumerate}
\end{procedure}

Without loss of generality, assume the resources in $\tS_{\pi_i}$ have
been sorted in decreasing order of congestion and vice versa for resources
in $\tS^*_{\pi_i}$, i.e $C_{r_1} \geq C_{r_2} \ldots \geq \ldots C_{r_m}$
and $C_{r^*_1} \leq C_{r^*_2} \ldots \leq \ldots C_{r^*_m}$.
Then we have the following:

\begin{lemma} 
There exists an implementation of $\mathrm{PMS\-Partition}(\pi_i)$ in
which 
\begin{enumerate}
\item
The $L^*_j$'s, $1 \leq j \leq t$, form a linear partition of
$\tS^*_{\pi_i}$ into contigous resources with $|L^*_j \bigcap L^*_{j+1}|
\leq 1$. If $|L^*_j \bigcap L^*_{j+1}| = 1$ then the last resource in
$L^*_j$ is the first resource in $L^*_{j+1}$.
\item
$\forall j: 1 \leq j \leq t$, either $|L_j| = 1$  or  $|L_j^*| =1$ or both.
If $|L_j| > 1$ and $|L^*_j|=1$ with $L^*_j = \{ r^*_p \}$
we must have $C_{r^*_p} \geq \max \{ C_r | r \in L_j \}$.
\end{enumerate}
\label{PMSPc}
\end{lemma}
\begin{proof} 
We provide a simple proof sketch due to space limitations.  Start with $L_1
= \{r_1\}$. We add resources $r^*_1, \ldots, r^*_q$ to $L^*_1$ where
$r^*_q$ is the first resource such that $\sum_{j=1}^q (C_{r^*_j}+1)^{\M}
\geq C_{r_1}^{\M}$.  Then we proceed with $L_2 = \{ r_2 \}$ and start
forming $L^*_2$ with $r^*_q$. As we continue this process,
due to the fact that 
\beq
\sum_{j=1}^m  (C_{r^*_j}+1 )^{\M}  \geq \sum_{l=1}^k C_{r_l}^{\M},
\label{equileq}
\eeq
eventually resources in $L$ will have smaller congestion than
the resources in $L^*$. At this point the $L^*$ partitions will
contain single resources while the corresponding $L$ partition will
contain multiple resources. At each step, we maintain the invariant
in Equation~\ref{sumres1}, which implies condition 2 in the lemma.
No resource in $L^*_j$ need be used more than twice during this process,
which is assured because of Eq.~\ref{equileq}.  We skip the remaining
technical details of the proof which ensure that as many resources as
possible from $\tS^*_{\pi_i}$ are used in the partition-pairs.
\end{proof}

Procedure $\mathrm{PMS\-Partition}()$ is used to create new players and
forms the basic step in our transformation algorithm.  We ensure the
equilibrium of these new players in $\tG$ using the key constructs of
{\it exact matching sets} and {\it potential matching sets}.

A set of resources $\tilde{R}$ in $\tG$ forms an exact matching set
for a newly created player $\pi_k$ with newly assigned equilibrium
strategy $\tS_{\pi_k}$ if $\sum_{r \in \tilde{R}} (C_r +1)^{\M} \geq
pc_{\pi_k}(\tS_{\pi_k},\tS_{-\pi_k}) = \sum_{r \in \tS_{\pi_k}}
C_r^{\M}$. Clearly, $\tilde{R}$ can be assigned as the new optimal
strategy $\tS^*_{\pi_k}$ in game $\tG$ without violating the equilibrium
of $\pi_k$.

Potential matching sets are defined for newly created type-$B$ players.
A potential matching set $\tilde{R}$ is an exact matching set that can
`potentially' be added to the optimal set of resources $\tS^*_{\pi_k}$
of a type-$B$ player $\pi_k \in \tG$ without increasing the optimal
bottleneck congestion in $\tG$ from original game $G$ by a constant
factor i.e $\tilde{C}^* \leq \beta C^*$, where $\beta >1$ is a constant.

Now consider a type-$B$ player $\pi_i$ to be transformed. We partition
the resources in its equilibrium and optimal strategies $\tS_{\pi_i}$
and $\tS^*_{\pi_i}$ according to $\mathrm{PMS\-Partition}(\pi_i)$
and remove it from $\tG$, i.e $\tS = \tS -\tS_{\pi_i}$ and $\tS^* =
\tS^* - \tS^*_{\pi_i}$.  

Consider those partition-pairs $(L_j,L^*_j)$ with $|L_j|=1$. We can
create a new type-$A$ player $\pi_k$ and add it to to $\tG$ with an
equilibrium strategy $\tS_{\pi_k}$ that is the singleton resource in
$L_j$.  Due to the condition in Eq.~\ref{sumres1}, the set of resources
in $L^*_j$ forms an exact matching set for $\pi_k$ and can therefore
be assigned to $\tS^*_{\pi_k}$. $\pi_k$ is in equilibrium in $\tG$ and
the equilibrium and optimal congestion on resources in $\tS_{\pi_k}$
and $\tS^*_{\pi_k}$ are now the same as before. This forms the `easy'
part of the transformation process.

Consider however, those partitions $(L_j,L^*_j)$ with $1 < |L_j| \leq |R|$
and $L^*_j = \{ r^*_l \}$.  Similar to the above, we can create $|L_j|$
new type-$A$ players and assign a distinct resource in $L_j$ to each such
players equilibrium strategy.  However if, as above, we assign $r^*_l$,
the single resource in $L^*_j$, to each players optimal strategy, we
might increase the socially optimal congestion $\tilde{C}^*$ of $\tG$
to as much as $ C^* + |R|$, thereby violating the domination of $G$
over $\tG$.  Thus we need to find an appropriate potential matching set
from among existing resources and assign them to these players, without
increasing the optimal congestion beyond $O(C^*)$.  Finding such a set
is the `hard' part of the transformation process and forms the core of
our algorithm below.

We define a subroutine (PARTITION-TRANSFORM)that executes procedure
$\mathrm{PMS\-Partition}(\pi_i)$ for a type-$B$ player $\pi_i$ and
creates several new type-$A$ and type-$B$ players.  Specifically, we
first obtain $\mathrm{PMS\-Partition}(\pi_i) = (L_1, L^*_1), \ldots,
(L_t, L^*_t)$.  We then delete the strategies of $\pi_i$ from $\tG$
and transform $\pi_i$ into $t$ new type-$A$ and type-$B$ sub-players
as follows:. for each partition member $L_q$, we create a new player
$\pi_{L_q}$ which is either a type-$A$ player if $|L_q|=1$, or a type-$B$
player if $|L_q|>1$, $1 \leq q \leq t$. $\pi_{L_q}$ is created with two
strategy sets: equilibrium strategy $\tS_{\pi_{L_q}} = L_q$ and socially
optimal strategy $\tS^*_{\pi_{L_q}} = L^*_q$.

The following lemma is a direct consequence of lemma~\ref{PMSPc} and is
needed for later analysis.

\begin{lemma}
For a given type-$B$ player $\pi_i$, every new type-$B$ player $\pi_{L_q}$
created after an execution of subroutine PARTITION-TRANSFORM($\pi_i$)
is in equilibrium in $\tG$ with $pc_{\pi_{L_q}} \leq (C_p+1)^{\M} \leq
pc_{\pi_i}$, where $C_p = \max{C_r| r \in \tS^*_{\pi_i}}$. Every new
type-$A$ player $\pi_q$ is also in equilibrium with $pc_{\pi_{L_q}}
\leq pc_{\pi_i}$.
\label{usefullemma1}
\end{lemma}



We are now ready to prove our main result.  
\bigskip

{\it Proof of Theorem~\ref{BtoAlemma}}:
We describe the transformation via an iterative algorithm for which
the pseudocode is attached below.  The main challenge is to find
potential matching sets for newly created type-$B$ players without
increasing the optimal congestion in the game beyond a constant factor.
To achieve this, we will transform type-$B$ players in distinct phases
corresponding to decreasing ranges of player costs.  As a preprocessing
step in the algorithm we call PARTITION-TRANSFORM($\pi$) for all type-$B$
players $\pi$ with $pc_{\pi} > (C+1)^{\M}$. By lemma~\ref{usefullemma1},
we are now left only with players with cost $\leq (C+1)^{\M}$ in $\tG$.


\begin{algorithm}
\caption{TRANSFORMATION ALGORITHM}
\label{alg1}
\begin{algorithmic}[1]
\STATE {\bf Preprocessing}: 
\STATE $\forall$ Type-$B$ players $\pi$ with $pc_{\pi} > (C+1)^{\M}$ 
\STATE \ \ Execute PARTITION-TRANSFORM($\pi$)
\STATE {\bf Main Procedure}
\FOR{$i = 1$ to $C+1-\psi$}
\STATE Phase Index $\wc = C+1-i$
\STATE $\Pi_{\wc} \leftarrow \{\mbox{type-}B \pi_j \in \tP| \wc^{\M} < pc_{\pi_j}(\tS_{\pi_j}) \leq (\wc+1)^{\M} \}$
\FORALL{$\pi_l \in \Pi_{\wc}$} 
\STATE PARTITION-TRANSFORM($\pi_l$)
\ENDFOR
\STATE $\Pi_{\wc} \leftarrow \{\mbox{type-}B \pi_j \in \tP| \wc^{\M} < pc_{\pi_j}(\tS_{\pi_j}) \leq (\wc+1)^{\M} \}$
\STATE $\Pi_{\wc} \leftarrow \Pi_{\wc} \bigcup \{\mbox{type-}A$
players with cost $= \wc^{\M}$
\FORALL {$\pi_l \in \Pi_{\wc}$} 
\STATE ELIMINATE-HIGH-CONGESTION-RESOURCES($\pi_l$) 
\ENDFOR
\STATE $D \leftarrow \{\pi_l\}$ where 1)$|\tS^*_{\pi_l}| > 1$,  2)$\max_{r \in \tS^*_{\pi_l}} \{C_r\} \leq \wc-1$
\STATE and 3) $\sum_{r \in \tS^*_{\pi_l} } (C_r+1)^{\M} \geq (\wc+1)^{\M}$ 
\FORALL {$\pi_l \in D$} 
\STATE PARTITION-TRANSFORM($\pi_l$)
\ENDFOR
\STATE $E \leftarrow \{\pi_l\}$ where $|\tS^*_{\pi_l}| = 1$ and $C_{\tS^*_{\pi_l}} = \wc$
\STATE $X \leftarrow \{ r \in R| C_r = \wc  \}$
\WHILE {$|E| > 0$}
\STATE Choose any $\pi_l \in E$ and select {\bf UNMARKED} 
\STATE type-$A$ player $\pi_j \in \Pi_X$ in round-robin fashion.
\STATE Update $\tS^*_{\pi_l} \leftarrow \tS^*_{\pi_l} \bigcup \tS^*_{\pi_j}$
\STATE Update $\tS^*_{\pi_j} \leftarrow \tS_{\pi_j}$ and {\bf MARK} $\pi_j$
\STATE PARTITION-TRANSFORM($\pi_l$) and add the resultant
new player to $E$ if qualified
\ENDWHILE
\ENDFOR 
\end{algorithmic}
\end{algorithm}

Let $\wc = C+1-i$ denote a (decreasing) phase index. During the $i^{th}$
phase, $1 \leq i \leq C + 1 - \psi$, we transform all type-$B$ players
with player costs $\wc^{\M} <  pc_{\pi} \leq (\wc+1)^{\M}$ into either
type-$A$ players or type-$B$ players of cost $\leq \wc^{\M}$.  We will
use the fact that in phase $i$, all resources with congestion $\geq
\wc+1$ in equilibrium state $\tS$ are occupied only by type-$A$ players
(since any type-$B$ player using a resource $r$ with $C_r \geq \wc+1$
would have a player cost strictly $ > (\wc+1)^{\M}$, a contradiction).

In phase $i$, let $\Pi_{\wc}$ denote the set of type-$B$ players
$\pi_j$ whose player costs are in the range $\wc^{\M} < pc_{\pi_j}
\leq (\wc+1)^{\M} $. We first call on PARTITION-TRANSFORM for all
players in $\Pi_{\wc}$.  This results in a new set of type-$A$ and
type-$B$ players with the same or lower player costs.  New type-$B$
players with cost $\leq \wc^{\M}$, are dealt with in subsequent phases
while we form $\Pi_{\wc}$ again with the remaining type-$B$ players.
At this point, every type-$B$ player $\pi_l \in \Pi_{\wc}$ has exactly
one resource in its optimal strategy (by definition of lemma~\ref{PMSPc}).
Moreover, this resource must have congestion $C_r
\geq \wc$ in equilibrium state $\tS$, since $pc_{\pi_l}(\tS) > \wc^{\M}$.

As discussed before, the key challenge is to find a larger potential
matching set for such type-$B$ players without increasing $\tC^*$
significantly.  For technical reasons, we first eliminate all
high-congested resources from consideration as potential matching sets.
In particular, resources with equilibrium congestion $C_r \geq \wc+1$
are occupied only by type-$A$ players. By eliminating these resources
(using subroutine ELIMINATE-HIGH-CONGESTION-RESOURCES($\pi_l$) described
below), we ensure that the optimal congestion $\tilde{C}_r^*$ on any
resource $r$ with equilibrium congestion $C_r$ remains unchanged during
all phases with phase index $\wc < C_r$.


Let $\pi_l$ denote a generic player from the set of type-$B$ players
in $\Pi_{\wc}$ and the set of type-$A$ players in $\tG$ with player
cost exactly $\wc^{\M}$ {\it and} a single resource in their optimal
strategy sets.  Let $\tS^*_{\pi_l} = \{ x \}$.  We check to see if $C_x
\geq \wc+1$.  If so, we find the type-$A$ player $\pi_q \in \Pi_x$
(recall that $\Pi_x$ is the set of players using $x$ in equilibrium)
with the largest socially optimal strategy set $|\tS^*_{\pi_q}|$. Let
$F = \argmin_{C_r \geq \wc} \{r \in \tS^*_{\pi_q}\} $, i.e the resource in
$\tS^*_{\pi_q}$ with the smallest congestion $\geq \wc$. If $F$ above
does not exist, then set $F = \tS^*_{\pi_q}$.  We now change the socially
optimal strategy of $\pi_l$ to $F$ instead of $x$. Since $\sum_{y \in F}
(C_y+1)^{\M} \geq (\wc+1)^{\M} \geq pc_{\pi_l}$, we are assured that
$\pi_l$ will remain in equilibrium in $\tS$ after this. Simultaneously,
change the optimal strategy of $\pi_q$ to its equilibrium strategy
i.e $\tS^*_{\pi_q} = \tS_{\pi_q} = x$.  Note that after this step,
the equilibrium and optimal congestion on all resources involved, i.e
$x$ and $F$, remain unchanged in $\tG$.  If the resource $x$ above
has congestion $C_x > \wc$, we repeat the steps for player $\pi_l$.
We execute the subroutine for all such qualified players.

It must now be the case that the set of players $\{ \pi_l \}$ can be
divided into two subsets $D$ and $E$, where $D$ contains all the players
with $|\tS^*_{\pi_l}| > 1$,  $\max_{r \in \tS^*_{\pi_l}} \{ C_r \} \leq \wc-1$
and  $\sum_{r \in \tS^*_{\pi_l} } (C_r+1)^{\M} \geq (\wc+1)^{\M}$
while $E$ contains players with $|\tS^*_{\pi_l}| = 1$ and $C_y = \wc$,
where $\tS^*_{\pi_l} = \{y\}$.  We now execute PARTITION-TRANSFORM on
all type-$B$ players in the set $D$. By lemma~\ref{PMSPc}, the cost of
a newly created type-$B$ player after PARTITION-TRANSFORM on the set $D$
can be at most $\wc^{\M}$. These players will be further transformed in
subsequent phases.

It now only remains to transform type-$B$ players from the set $E$
in this phase.  Let $X = \{ r \in R| C_r = \wc  \}$ denote the set of
resources with congestion exactly $\wc$ in equilibrium. Let $\Pi_X$
denote the set of players using resources from $X$ in equilibrium.
From the above discussions, we note the following:
\begin{enumerate}
\item
Every type-$B$ and type-$A$ player $\pi_l \in E$ is using a
resource from $X$ in its optimal strategy, i.e $\tS^*_{\pi_l} = x$
for some $x \in X$.

\item
Every type-$A$ player in $D \bigcup E$ is using a resource from $X$ in
its equilibrium strategy.  
\end{enumerate}

Essentially every untransformed type-$B$ player is using a resource in
$X$ in its optimal strategy.  Let $X^B \subseteq E$ denote the set of
type-$B$ players in $\Pi_X$.  Then $X^A = \Pi_X - X^B$, the remaining set of
players in $\Pi_X$, must be of type-$A$.  We can focus on the set $X^A$
to obtain larger potential matching sets for type-$B$ players in $E$.

Order the resources in $X$ as $L = r_1, r_2, \ldots r_k$, in increasing
order of number of type-$A$ players using them in equilibrium,
where $k \leq |X|$.  For each type-$B$ player $\pi_l \in E$, we
select an {\it unmarked} type-$A$ player $\pi_j$ using a resource,
say $r_m$ in $L$, (i.e $\tS_{\pi_j} = r_m$) where $\tS^*_{\pi_j}
\bigcap \tS^*_{\pi_l} = \emptyset$.  Now define $PMS(\pi_l) = \tS^*_{\pi_j}
\bigcup \tS^*_{\pi_l}$ as the new potential matching set for $\pi_l$.
Update the optimal strategy of $\pi_l$ to $\tS^*_{\pi_l} = PMS(\pi_l)$
and execute PARTITION-TRANSFORM($\pi_l$).  We claim that at least two
new players are created, at most one of which could be a type-$B$ player
with cost $> \wc^{\M}$.  This is because $\sum_{r \in \tS^*_{\pi_j}}
(C_r+1)^{\M} \geq \wc^{\M}$ while every resource in $\tS_{\pi_l}$ has
congestion $\leq \wc$. Thus resources from $\tS^*_{\pi_j}$ will be in at
least one partition pair of $\mathrm{PMS\-Partition}(\pi_l)$ while the
(existing) resource in $\tS^*_{\pi_l}$ will be in another partition pair. If after
$\mathrm{PMS\-Partition}(\pi_l)$ there is still a type-$B$ player with
cost $>\wc^{\M}$ add this player to the set $E$.  Simultaneously update
the optimal strategy of $\pi_j$ to $\tS^*_{\pi_j} = r_{m}$ and {\it mark}
$\pi_j$. (Note that this increases the optimal congestion on $r_m$ by one.
We will bound the total increase in optimal congestion later). Also note
that all players are in equilibrium in $\tS$.  We repeat the process as
long as there exist type-$B$ players in set $E$.

In order to show that this is a valid transformation, we need to show
that there exist a sufficient number of unmarked type-$A$ players in
$X$.  We do this by a simple counting argument.  Let $\alpha = \lfloor
e^{\frac{\M}{\psi}} \rfloor $. First note that each type-$B$ player
in $X^B$ can be using at most $\lfloor (\wc+1)^{\M}/\wc^{\M}\rfloor \leq
\alpha$ resources in $X$ in equilibrium state $\tS$ and hence $|X^A| \geq
|X| \cdot \wc - \alpha |X^B|$.  Secondly, since each resource in $X$ can
be in the optimal strategy of at most $C^*$ type-$B$ players and every
type-$B$ player in $E$  at the start of the transformation process
has its optimal strategy in $X$, we must have $|X^B| \leq C^* |X|$.
Finally, noting that each type-$B$ player in $E$ at the start of the
transformation process can make at most $2 \alpha$ calls for marking
type-$A$ players before it is completely transformed we get the total
number of type-$A$ players in $X$ required for marking as $\leq 2 \alpha
C^* |X|$.  Using the given threshold value $\psi = \max(2\M,3C^*)$, we obtain
the number of type-$A$ players in $X$ as

\beq 
|X^A| \geq |X| \cdot \wc - |X^B| \geq |X|(\wc - C^*) \geq 2 C^* |X| 
\eeq
for $\wc > \psi \geq 3 C^*$ which proves the result.
$\hfill \Box$


\medskip

{\it Proof of Theorem~\ref{dominatedtheorem}}: First note that the optimal congestion
$\tilde{C}^*_r$ on a resource $r$ does not change in any phase with phase index
$\wc < C_r$.  There are only two occasions when $\tilde{C}^*_r$ increases:

\begin{enumerate}
\item
During the phase with index $\wc = C_r$, $\tilde{C}^*_r$ increases
by one whenever a type-$A$ player on $r$ is marked.  The number of
resources in $X$ that contain type-$A$ players to be marked is $|R^A|
\geq |X^A|/\wc \geq |X|\frac{(\wc - C^*)}{\wc}$.  In order to bound
the increase in $\tilde{C}^*_r$, we will select (and mark) type-$A$
players from $R^A$ during each step of the transformation in {\it cyclic
round-robin fashion}, i.e after marking $r_m$, we select a type-$A$
player from $r_{m+1}$ etc. Since at most $2 C^* |X|$ players are required to be
marked, the maximum number of marked players
on any resource in $X$ is $\leq 2C^*|X|/|R^A|$ which is bounded by
\[
\frac{2 \wc C^*}{\wc -C^*} \leq  \frac{2 \psi C^*}{\psi -C^*} \leq  3C^*
\]

Hence for any resource $r$ optimal congestion $\tilde{C}^*_r$ increases by
at most $3C^*$ in phase index $C_r$ due to marked players.

\item

At the beginning of the transformation process, $r$ can be in the
optimal strategy set, and consequently partition pairs of up to $C^*$
players. In the discussion below, focus on a particular player in this set
and consider the increase in $\tilde{C}^*_r$ throughout the transformation
process only due to $r$'s presence in some partition pair $L^*_k$ due to
this particular player.  $r$ via $L^*_k$ can be involved in multiple calls
to $\mathrm{PMS\-Partition}()$ in multiple phases with phase index $>
C_r$. In any such call to $\mathrm{PMS\-Partition}()$, $\tilde{C}^*_r$
can increase by 1 only if $r$ is either the first or last resource in
a new partition pair $L_j^* \subset L^*_k$.  Once $r$ is the first or
last resource in a partition $L^*_j$, its $\tilde{C}^*_r$ can increase
by at most one when $r$ become part of a new singleton partition pair
(i.e $ L^*_k = \{r\}$). From this point onwards, $\tilde{C}^*_r$ cannot
increase due to further $\mathrm{PMS\-Partition}()$ calls.  Thus the total
increase in $\tilde{C}^*_r$ is bounded by 3.  Given that $r$ can be in up
to $C^*$ optimal strategy sets at the start of the first phase, the total
increase in $\tilde{C}^*_r$ due to calls to $\mathrm{PMS\-Partition}()$
is bounded by $3C^*$.

\end{enumerate}

Putting the two facts above together, consequently, we get $\tC^* \leq
C^* + 3C^* + 3C^* = 7C^*$ and hence $\tG$ is dominated by $G$.
$\hfill \Box$

\section{Price of Anarchy}
\label{section:anarchy}

\subsection{Price of Anarchy for Type-A Player Games}

We now consider equilibria where highly congested resources are occupied
only by type-$A$ players and use this to bound the price of anarchy
of games with polynomial cost functions.  Consider a game with optimal
solution $S^* = (S^*_{\pi_1}, \ldots, S^*_{\pi_n})$ and congestion $C^*$.
Let $\psi = \max(2M,3C^*)$ be a threshold value.  Let $S = (S_{\pi_1},
\ldots, S_{\pi_n})$ denote the Nash equilibrium state which has the
highest congestion $C$ among all Nash equilibria states, and further
all players on resources $r$ with $C_{r} > \psi$ are of type-A.
We will obtain a price of anarchy result by bounding the ratio $C/C^*$.

We first define a resource graph $\H$ for state $S$.  There are $V= V_1
\bigcup V_2$ nodes in $\H$. Each resource $r \in R$ with $C_r > \psi$
($C_r \leq \psi$, resp.)  corresponds to the equivalent node $r \in V_1$
($r \in V_2$).  Henceforth we will use the term resource and node
interchangeably.  For every player $\pi$ using a resource $x \in V_1$
in equilibrium, there is a directed edge $(x,y)$ between node $x$ and
all nodes $y \in V$, where $y \neq x$ is in the optimal strategy set of
$\pi$ i.e $S_{\pi} = x$ and $y \in S^*_{\pi}$.  The set of nodes in $
\bigcup_{\pi: S_{\pi} = x} S^*_{\pi}$ are called the children of node
$x$. We use the notation $\mathrm{Ch}(x)$ to denote this set. 
Note that there could be multiple links directed at $x$ from the same
node, however $x$ can be the child of at most $C^*$ nodes and $x$ cannot
be its own child.  Also note that nodes in $V_2$ are terminal nodes that
have no outgoing links.

We first observe the following about nodes in $V_1$:
\begin{lemma}
\label{lemma:H-expansion}
For any node $x \in \H$ with $C_{x} > \psi$,
it holds that 
\[
\sum_{y \in \mathrm{Ch}(x) \bigcap V_1} C_y^{\M} + \sum_{y \in \mathrm{Ch}(x) \bigcap V_2} \psi^{\M} 
\geq \frac{C_x - C^*}{2C^*} C_x^{\M}
\]
\end{lemma}
\begin{proof}
Let $\Pi$ be the set of players such that $\forall \pi \in \Pi: S_{\pi}
= x$ and $x \not \in S^*_{\pi}$.  We must have $C_{x}-C^* \leq |\Pi|
\leq C_{x}$ since up to $C^*$ players could be using resource $x$
simultaneously in their optimal as well as equilibrium strategies.
Since $\pi$ is in equilibrium state, we must have $\sum_{ y \in S^*_{\pi}
} (C_y + 1)^{\M} \geq  C_x^{\M}$.  Let $Z_{\pi} = S^*_{\pi} \bigcap V_1$ and
$W_{\pi} =  S^*_{\pi} \bigcap V_2$.  Using the fact that $\forall z \in Z_{\pi}: C_z
> \psi $ and $C^* \geq 1$, we get that $((C_z+1)/C_z)^{\M} \leq (1 +
\frac{1}{\psi})^{\M} \leq \sqrt{e} < 2$.
Also $\forall w \in W: C_w +1 \leq \psi+1 < 2 \psi$ and hence $\sum_{z
\in Z_{\pi}} 2 (C_z)^{\M} + \sum_{w \in W_{\pi}}  2 \psi^{\M} \geq C_x^{\M}/2$.
Now consider the sum
\[
\sum_{\pi \in \Pi} \left( \sum_{z \in Z_{\pi}} C_z^{\M} + \sum_{w \in W_{\pi}}  \psi^{\M} \right)
\]
Since $|\Pi| \geq C_x-C^*$ and each resource $z$ and $w$ in the inside term above
can be in the sets $Z_{\pi}$ and $W_{\pi}$ for up to $C^*$ players from $\Pi$, we get
\beq
\sum_{y \in \mathrm{Ch}(x) \bigcap V_1} C_y^{\M} + \sum_{y \in \mathrm{Ch}(x) \bigcap V_2} \psi^{\M} 
\geq \frac{C_x - C^*}{2C^*} C_x^{\M}
\label{Childxlemmaeq}
\eeq
as desired.
\end{proof}

To get a bound on the price of anarchy we need to relate the number
of resources $|R|$ with the parameters $C$ and $C^*$. Note that the
second term on the LHS of Eq.~\ref{Childxlemmaeq} is bounded by $|R|
\psi^M$. Unfortunately, since $\H$ has cycles, we cannot apply the lemma
recursively to nodes from the first term on the LHS and their children
in $\H$ (to eventually replace these nodes with nodes from $V_2$).
We will therefore modify $\H$ to eliminate cycles and construct an {\em
expansion } Directed Acyclic Graph (DAG) $\T$ (without increasing the
size of $\H$), which will help us obtain our price of anarchy bound.
This is stated in the form of the lemma below (We omit the proof due to
space considerations).

\begin{lemma}
\label{lemma:HtoT}
Resource graph $\H$ can be transformed into expansion DAG $\T$ 
without affecting the equilibrium state $S$ and optimal
congestion $C^*$, where $|\T| \leq |\H|$. 
\end{lemma}

Since $\T$ is a DAG we know that it has sink nodes (with outdegree
0). Every node in $V_1$ is an internal node (with non-zero indegree
and outdegree) since it has congestion $> C^*$ and hence the sink nodes
in $\T$ are nodes from $V_2$.  Consider the DAG starting at the root
node with congestion $C$. By applying lemma~\ref{lemma:H-expansion}
recursively to the root and its descendants in $V_1$, we  count the
number of nodes until we reach terminating sinks in $V_2$. Noting as
before that each resource in $T$ is counted at most $C^*$ times in the
lemma, we get the result

\begin{lemma}
For DAG $\T$ with root node $r$ and congestion $C_r =C$, it holds that
\[
\sum_{y \in \mathrm{Descendants}(r) \bigcap V_2} \psi^{\M} 
\geq \frac{C - C^*}{2C^*} C^{\M}
\]
\label{descendantcount}
\end{lemma}

\begin{theorem}[Price of Anarchy for Type-A players]
\label{PoA-A}
The upper bound on the price of anarchy is $PoA = O(|R|^{\frac{1}{\M+1}})$.
\end{theorem}
\begin{proof}
The number of descendants of $r$ in $T$ is at most $|R|-1$.
Using the fact that $\psi \geq 3C^*$ and substituting in lemma~\ref{descendantcount},
we get that 
\[
|R|-1 \geq (PoA-1)\cdot \frac{PoA^{\M}}{2 \cdot 3^{\M}} > \frac{PoA^{M+1}}{4 \cdot 3^M}
\]
for any $PoA > 2$, and hence for the given constant $\M$, we get the desired result
$PoA = O(|R|^{\frac{1}{\M+1}})$.
\end{proof}

\subsection{Price of Anarchy for Arbitrary Games}
By Theorem~\ref{BtoAlemma}, we only need to consider games in equilibrium
with type-A players occupying resources with congestion $> \psi$.
By combining Theorem \ref{BtoAlemma}, Theorem \ref{PoA-A}, and Corollary
\ref{dominatingcorollary} we obtain the main result for price of anarchy:

\begin{theorem}[$PoA$ for Arbitrary Polynomial Cost Games]
\label{PoA-main}
The upper bound on the price of anarchy for polynomial cost games is
$O(|R|^{\frac{1}{\M+1}})$.
\end{theorem}


\section{Lower Bound}
\label{section:lower}

We show that the upper bound of $O(|R|^{1/(M+1)}$ in the price of anarchy is tight
by demonstrating a congestion game with a 
lower bound on the price of anarchy of $\Omega(|R|^{1/(M+1)})$.
We construct a game instance represented as a graph in the figure below,
such that each edge in the graph corresponds to a resource,
and each player $\pi_i$ has two strategies available: either the path from $u$ to $v$
through the direct edge $e = (u,v)$, or an alternative path $p_i = (u,x_i, \ldots , y_i, v)$
(note that different player paths are edge-disjoint).
Each path has length $|p_i| = |R|^{\M/(\M+1)}$ edges
and the number of players is $n = |R|^{1/(\M +1)}$.
(For simplicity assume that the values of $|p_i|$ and $n$ are integers,
since if they are not we can always round to the nearest ceiling.)
The edge $e$ is actually part of one the paths; for ease of presentation,
in the figure below the edge is depicted 
as separate from the paths. 
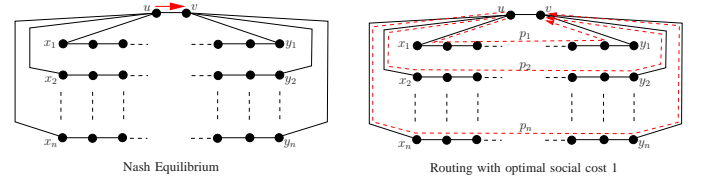
\begin{center}
\resizebox{3.5in}{!}{\begin{picture}(0,0)%
\includegraphics{sumc.pstex}%
\end{picture}%
\setlength{\unitlength}{3947sp}%
\begingroup\makeatletter\ifx\SetFigFont\undefined%
\gdef\SetFigFont#1#2#3#4#5{%
  \reset@font\fontsize{#1}{#2pt}%
  \fontfamily{#3}\fontseries{#4}\fontshape{#5}%
  \selectfont}%
\fi\endgroup%
\begin{picture}(10092,2685)(270,-2477)
\put(905,-502){\makebox(0,0)[rb]{\smash{{\SetFigFont{12}{14.4}{\rmdefault}{\mddefault}{\updefault}$x_1$}}}}
\put(4363,-502){\makebox(0,0)[lb]{\smash{{\SetFigFont{12}{14.4}{\rmdefault}{\mddefault}{\updefault}$y_1$}}}}
\put(2331, 27){\makebox(0,0)[rb]{\smash{{\SetFigFont{12}{14.4}{\rmdefault}{\mddefault}{\updefault}$u$}}}}
\put(2937, 36){\makebox(0,0)[lb]{\smash{{\SetFigFont{12}{14.4}{\rmdefault}{\mddefault}{\updefault}$v$}}}}
\put(923,-2005){\makebox(0,0)[rb]{\smash{{\SetFigFont{12}{14.4}{\rmdefault}{\mddefault}{\updefault}$x_{n}$}}}}
\put(4338,-1996){\makebox(0,0)[lb]{\smash{{\SetFigFont{12}{14.4}{\rmdefault}{\mddefault}{\updefault}$y_{n}$}}}}
\put(923,-1057){\makebox(0,0)[rb]{\smash{{\SetFigFont{12}{14.4}{\rmdefault}{\mddefault}{\updefault}$x_2$}}}}
\put(4355,-1048){\makebox(0,0)[lb]{\smash{{\SetFigFont{12}{14.4}{\rmdefault}{\mddefault}{\updefault}$y_2$}}}}
\put(6251,-531){\makebox(0,0)[rb]{\smash{{\SetFigFont{12}{14.4}{\rmdefault}{\mddefault}{\updefault}$x_1$}}}}
\put(9709,-531){\makebox(0,0)[lb]{\smash{{\SetFigFont{12}{14.4}{\rmdefault}{\mddefault}{\updefault}$y_1$}}}}
\put(7677, -2){\makebox(0,0)[rb]{\smash{{\SetFigFont{12}{14.4}{\rmdefault}{\mddefault}{\updefault}$u$}}}}
\put(8283,  7){\makebox(0,0)[lb]{\smash{{\SetFigFont{12}{14.4}{\rmdefault}{\mddefault}{\updefault}$v$}}}}
\put(6269,-2034){\makebox(0,0)[rb]{\smash{{\SetFigFont{12}{14.4}{\rmdefault}{\mddefault}{\updefault}$x_{n}$}}}}
\put(9684,-2025){\makebox(0,0)[lb]{\smash{{\SetFigFont{12}{14.4}{\rmdefault}{\mddefault}{\updefault}$y_{n}$}}}}
\put(6269,-1086){\makebox(0,0)[rb]{\smash{{\SetFigFont{12}{14.4}{\rmdefault}{\mddefault}{\updefault}$x_2$}}}}
\put(9701,-1077){\makebox(0,0)[lb]{\smash{{\SetFigFont{12}{14.4}{\rmdefault}{\mddefault}{\updefault}$y_2$}}}}
\put(7985,-373){\makebox(0,0)[b]{\smash{{\SetFigFont{12}{14.4}{\rmdefault}{\mddefault}{\updefault}$p_1$}}}}
\put(7994,-1756){\makebox(0,0)[b]{\smash{{\SetFigFont{12}{14.4}{\rmdefault}{\mddefault}{\updefault}$p_n$}}}}
\put(7985,-834){\makebox(0,0)[b]{\smash{{\SetFigFont{12}{14.4}{\rmdefault}{\mddefault}{\updefault}$p_2$}}}}
\put(2637,-2402){\makebox(0,0)[b]{\smash{{\SetFigFont{12}{14.4}{\rmdefault}{\mddefault}{\updefault}Nash Equilibrium}}}}
\put(7964,-2413){\makebox(0,0)[b]{\smash{{\SetFigFont{12}{14.4}{\rmdefault}{\mddefault}{\updefault}Routing with optimal social cost 1}}}}
\end{picture}%
}
\end{center}
%
Let $S$ be the state depicted on the left part of the figure,
where each player chooses the first strategy,
and let $S^*$ be the state on the right part of the figure,
where each player chooses the alternative path.
We have that $C(S^*) = 1$, which is the smallest congestion possible.
Thus, $S^*$ represents a socially optimal solution.
For state $S$ we have that $C(S) = n$, since all players use edge $(u,v)$.
Note that $S$ is a Nash Equilibrium,
since each player $\pi_i$
has cost $pc_{\pi_i}(S) = n^{\M} = |R|^{\M/ (\M + 1)}$,
and the cost of switching to path $p_i$ would be
$1^{\M} \cdot |p_i|= |R|^{\M/ (\M + 1)}$,
which is the same at the cost of using edge $(u,v)$.
Consequently, a lower bound on the price of anarchy is
$C(S) / C(S^*) = n/1 = |R|^{1/(\M +1)}$.
Therefore, $PoA = \Omega(|R|^{1/(\M +1)})$, as needed.

\section{Conclusions}
We have considered bottleneck congestion games with polynomial
cost functions and shown that the Price of Anarchy is bounded by
$O(|R|^{\frac{1}{\M+1}})$. This Price of Anarchy result is optimal,
as demonstrated by a game with this exact $PoA$. We also demonstrate
two novel techniques, $B$ to $A$ player conversion and expansion which
help us obtain this result. These techniques which enable us to simplify
games for analysis are sufficiently general. In future work, we plan
to use these techniques to analyze the $PoA$ of games with 
arbitrary player cost functions.

{
\bibliographystyle{plain}
\bibliography{main}

\begin{thebibliography}{10}

\bibitem{BO07}
Ron Banner and Ariel Orda.
\newblock Bottleneck routing games in communication networks.
\newblock {\em IEEE Journal on Selected Areas in Communications},
  25(6):1173--1179, 2007.
\newblock also appears in INFOCOM'06.

\bibitem{BKV08}
Costas Busch, Rajgopal Kannan, and Athanasios~V. Vasilakos.
\newblock Quality of routing congestion games in wireless sensor networks.
\newblock In {\em Proc. 4th International Wireless Internet Conference
  (WICON)}, Maui, Hawaii, November 2008.

\bibitem{BM06}
Costas Busch and Malik Magdon-Ismail.
\newblock Atomic routing games on maximum congestion.
\newblock {\em Theoretical Computer Science}, 410(36):3337--3347, August 2009.

\bibitem{CK05}
George Christodoulou and Elias Koutsoupias.
\newblock The price of anarchy of finite congestion games.
\newblock In {\em Proceedings of the 37th Annual {ACM} Symposium on Theory of
  Computing (STOC)}, pages 67--73, Baltimore, {MD}, {USA}, May 2005. ACM.

\bibitem{correa1}
Jos{\'e}~R. Correa, Andreas~S. Schulz, and Nicol{\'a}s E.~Stier Moses.
\newblock Computational complexity, fairness, and the price of anarchy of the
  maximum latency problem.
\newblock In {\em Proc. Integer Programming and Combinatorial Optimization,
  10th International {IPCO} Conference}, volume 3064 of {\em Lecture Notes in
  Computer Science}, pages 59--73, New York, {NY}, {USA}, June 2004. Springer.

\bibitem{czumaj1}
Czumaj and Vocking.
\newblock Tight bounds for worst-case equilibria.
\newblock In {\em ACM Transactions on Algorithms (TALG)}, volume~3. ACM, 2007.

\bibitem{FKS02}
Dimitris Fotakis, Spyros~C. Kontogiannis, and Paul~G. Spirakis.
\newblock Selfish unsplittable flows.
\newblock {\em Theoretical Computer Science}, 348(2-3):226--239, 2005.

\bibitem{GLMMb04}
Martin Gairing, Thomas L{\"u}cking, Marios Mavronicolas, and Burkhard Monien.
\newblock Computing {Nash} equilibria for scheduling on restricted parallel
  links.
\newblock In {\em Proceedings of the 36th Annual {ACM} Symposium on the Theory
  of Computing (STOC)}, pages 613--622, Chicago, Illinois, {USA}, June 2004.

\bibitem{Harks09}
Tobias Harks, Max Klimm, and Rolf~H. M\"{o}hring.
\newblock Strong nash equilibria in games with the lexicographical improvement
  property.
\newblock In {\em WINE '09: Proceedings of the 5th International Workshop on
  Internet and Network Economics}, pages 463--470, Berlin, Heidelberg, 2009.
  Springer-Verlag.

\bibitem{SAGT}
Rajgopal Kannan and Costas Busch.
\newblock Bottleneck congestion games with logarithmic price of anarchy.
\newblock In {\em Proc. 3rd Annual Symposium on Algorithmic Game Theory (SAGT
  2010)}, October 2010.
\newblock To appear.

\bibitem{KMS02}
Elias Koutsoupias, Marios Mavronicolas, and Paul~G. Spirakis.
\newblock Approximate equilibria and ball fusion.
\newblock {\em Theory Comput. Syst.}, 36(6):683--693, 2003.

\bibitem{KP99}
Elias Koutsoupias and Christos Papadimitriou.
\newblock Worst-case equilibria.
\newblock In {\em Proceedings of the 16th Annual Symposium on Theoretical
  Aspects of Computer Science (STACS)}, volume 1563 of {\em LNCS}, pages
  404--413, Trier, Germany, March 1999. Springer-Verlag.

\bibitem{LMR94}
F.~T. Leighton, B.~M. Maggs, and S.~B. Rao.
\newblock Packet routing and job-scheduling in ${O}(congestion+dilation)$
  steps.
\newblock {\em Combinatorica}, 14:167--186, 1994.

\bibitem{libman1}
Lavy Libman and Ariel Orda.
\newblock Atomic resource sharing in noncooperative networks.
\newblock {\em Telecomunication Systems}, 17(4):385--409, 2001.

\bibitem{LMMR04}
Thomas L{\"u}cking, Marios Mavronicolas, Burkhard Monien, and Manuel Rode.
\newblock A new model for selfish routing.
\newblock {\em Theoretical Computer Science}, 406(3):187--206, 2008.

\bibitem{MS01}
Mavronicolas and Spirakis.
\newblock The price of selfish routing.
\newblock {\em Algorithmica}, 48, 2007.

\bibitem{monderer1}
D.~Monderer and L.~S. Shapely.
\newblock Potential games.
\newblock {\em Games and Economic Behavior}, 14:124--143, 1996.

\bibitem{P01}
Christos Papadimitriou.
\newblock Algorithms, games, and the {Internet}.
\newblock In {ACM}, editor, {\em Proceedings of the 33rd Annual {ACM} Symposium
  on Theory of Computing (STOC)}, pages 749--753, Hersonissos, Crete, Greece,
  July 2001.

\bibitem{rosenthal1}
R.~W. Rosenthal.
\newblock A class of games possesing pure-strategy {N}ash equilibria.
\newblock {\em International Journal of Game Theory}, 2:65--67, 1973.

\bibitem{roughgarden1}
Tim Roughgarden.
\newblock The maximum latency of selfish routing.
\newblock In {\em Proceedings of the Fifteenth Annual {ACM}-{SIAM} Symposium on
  Discrete Algorithms (SODA)}, pages 980--981, New Orleans, Louisiana, (USA),
  January 2004.

\bibitem{roughgarden2}
Tim Roughgarden.
\newblock Selfish routing with atomic players.
\newblock In {\em Proc. 16th Symp. on Discrete Algorithms (SODA)}, pages
  1184--1185. ACM/SIAM, 2005.

\bibitem{roughgarden3}
Tim Roughgarden and \'{E}va Tardos.
\newblock How bad is selfish routing.
\newblock {\em Journal of the ACM}, 49(2):236--259, March 2002.

\bibitem{roughgarden5}
Tim Roughgarden and \'{E}va Tardos.
\newblock Bounding the inefficiency of equilibria in nonatomic congestion
  games.
\newblock {\em Games and Economic Behavior}, 47(2):389--403, 2004.

\bibitem{STZ04}
Subhash Suri, Csaba~D. Toth, and Yunhong Zhou.
\newblock Selfish load balancing and atomic congestion games.
\newblock {\em Algorithmica}, 47(1):79--96, January 2007.

\end{thebibliography}
}




\end{document}